\documentclass[11pt]{lmcs}
\title{A static cost analysis for a higher-order language}

\def\wesaddr{Department of Mathematics and Computer Science \\
Wesleyan University \\
Middletown, CT 06459 USA}

\author{N.~Danner}
\address{\wesaddr}
\email{ndanner@wesleyan.edu}

\author{J.~Paykin}
\address{\wesaddr}
\email{jpaykin@wesleyan.edu}
\thanks{Current address for J.~Paykin:  
Department of Computer and Information Science,
University of Pennsylvania, Philadelphia, PA 19104 USA.}

\author{J.S.~Royer}
\address{Department of Electrical Engineering and Computer Science \\
Syracuse University \\
Syracuse, NY 13244 USA}
\email{jsroyer@syr.edu}

\makeatletter
\def\endfront@text{}
\makeatother

\usepackage{amsmath}
\usepackage{amsfonts}
\usepackage{listings}
\lstdefinestyle{mlcode}{
    language=ML,
    literate={->}{{${}\rightarrow{}$}}{2},
    basicstyle=\normalfont,
    keywordstyle=\texttt,
    morekeywords={int, bool, real, char, string, true, false, not, hd, tl},
    identifierstyle=\texttt,
    commentstyle=\texttt,
    stringstyle=\texttt,
    tabsize=2,
    showstringspaces=false,
    mathescape=true,
}
\usepackage{mathptmx}
\usepackage{microtype}
\usepackage{suffix}
\usepackage{stmaryrd}

\usepackage{url}
\usepackage{natbib}

\usepackage{bussproofs}

\usepackage[labelfont=bf]{caption}

\def\dom{\mathrm{Dom}\;}
\let\cross=\times
\def\setseparator{\mid}
\newcommand{\set}[2][\relax]{
  \ifx#1\relax
    \{#2\}
  \else
    \ifx#1\left
      \left\{#2\right\}
    \else
      \csname #1l\endcsname\{#2\csname #1r\endcsname\}
    \fi
  \fi
}
\newcommand{\setst}[3][\relax]{\set[#1]{#2\setseparator#3}}

\def\arrow{\mathbin{\rightarrow}}
\def\llambda{{\lambda\hskip-.45em\lambda}}	
\def\fv{\mathop{\mathrm{fv}}\nolimits}

\def\trans#1{\lVert{#1}\rVert}

\def\Exp{\mathrm{Exp}}

\def\nilexp{\lstinline!nil!}
\def\ifexp#1#2#3{\lstinline!if $\;#1\;$ then $\;#2\;$ else $\;#3$!}
\def\caseexp#1#2#3#4#5{\lstinline!case $\;#1\;$ of ($#2, [#3, #4]#5$)!}
\def\foldexp#1#2#3#4#5#6{\lstinline!fold $\;#1\;$ of ($#2, [#3, #4, #5]#6$)!}

\def\TT{\mathit{tt}}
\def\FF{\mathit{ff}}

\let\evalto\downarrow
\def\oftype{\mathbin{:}}
\let\proves\vdash

\def\cost{\mathop{\mathrm{cost}}\nolimits}
\def\pot{\mathop{\mathrm{pot}}\nolimits}
\def\dally{\mathop{\mathrm{dally}}\nolimits}

\EnableBpAbbreviations
\def\doRule#1{\ifx#1\relax\else\RightLabel{#1}\fi}
\newcommand{\ndAXC}[2][\relax]{\AXC{$\mathstrut$}\doRule{#1}\UIC{#2}}

\newcommand{\ndTIC}[2][\relax]{\doRule{#1}\TIC{#2}}

\def\eval#1#2{{#1}\evalto{#2}}

\def\typenohyp#1#2{{#1}\oftype{#2}}
\def\typehyp#1#2#3{{#1}\proves\typenohyp{#2}{#3}}
\newcommand{\typing}[3][\relax]{\ifx#1\relax\typenohyp{#2}{#3}\else\typehyp{#1}{#2}{#3}\fi}
\newcommand{\typingG}[2]{\typing[\Gamma]{#1}{#2}}

\newcommand{\typingE}[2]{\typing[\underbar{~}]{#1}{#2}}

\def\cl#1#2{{#1}\,#2}
\WithSuffix\def\cl*#1#2{({#1})#2}

\def\emptyenv{\{\}}
\def\extend#1#2#3{#1\{#2\mapsto#3\}}

\def\typot#1{\langle\!\langle#1\rangle\!\rangle}
\let\cpxy\trans

\def\N{\mathbf{N}}

\def\baseb{\mathsf{b}}

\let\bmax\vee
\def\pcaseexp#1#2#3#4#5{\lstinline!pcase $\;#1\;$ of ($#2, [#3, #4]#5$)!}
\def\pfoldexp#1#2#3#4#5#6{\lstinline!pfold $\;#1\;$ of ($#2, [#3, #4, #5]#6$)!}

\def\pcpxy#1{#1^{\Box}}
\def\parrow#1#2{{#1}\arrow^{\Box}{#2}}

\def\rec{\mathit{rec}}

\let\bounded\sqsubseteq
\def\vbounded{\bounded^{\mathrm{val}}}

\def\den#1#2{\llbracket{#1}\rrbracket{#2}}

\usepackage{hyperref}

\begin{document}

\titlecomment{%
ACM, 2013. This is the authors' version of the work. It is posted
here by permission of ACM for your personal use. Not for redistribution.
The definitive version was published in 
\emph{Programming Languages meets Program Verification 2013}.
}



\begin{abstract}
We develop a static complexity analysis for a
higher-order functional language with structural list recursion. 
The complexity of an expression is a pair consisting of a cost
and a potential.  The former is defined to be the 
size of the expression's evaluation derivation in a standard
big-step operational semantics.
The latter is a measure of the ``future'' cost of using
the value of that expression.
A translation function $\trans{\cdot}$ maps target 
expressions to complexities. Our main result is the following
Soundness Theorem: If $t$ is a term in the target language, then the cost 
component of~$\trans{t}$ is an upper bound on the cost of evaluating $t$.
The proof of the Soundness Theorem is formalized in Coq, providing
certified upper bounds on the cost of any expression in the target language.
\end{abstract}

\maketitle

\section{Introduction}

Though cost analyses are well-studied, they are traditionally
performed by hand in a relatively ad-hoc manner.
Formalisms for (partially) automating the analysis of higher-order
functional languages have been developed by, e.g.,
\citet{shultis:complexity},
\citet{sands:thesis},
\citet{van-stone:thesis}, and \citet{benzinger:tcs04}.\footnote{This
is an incomplete list; we focus here only
on work that directly addresses the analysis of higher-order 
functional languages and/or the automation of that analysis, and
we discuss these systems in more detail in Section~\ref{sec:related_work}.}  
These formalisms map target-language programs
into a domain of complexities, which can then be reasoned about
more-or-less formally.  The translations carry over information regarding
the values of subexpressions, which can then be ``discarded'' during
reasoning about cost should that be appropriate.

In this paper we aim for a similar goal, but take an approach
inspired by the work
of \cite{danner-royer:two-algs}.  There we analyzed a programming
language for type-level~$2$ functions in a restricted type system
motivated by work in implicit computational complexity.
The goal there was to prove that all programs in our formalism are
computable in type-$2$ polynomial time.  Here we take the same
analysis tools and apply them to a version of G\"odel's System~$T$,
not to establish fixed time bounds, but to construct expressions
that bound the cost of the given program.
The key difference between our approach and those discussed above is that
we aim for upper bounds on cost in terms of input
size, rather than an exact analysis
in terms of values.
This more modest goal allows us to develop a notion of complexity
that bounds the run-time cost of any target program evaluation
on inputs of a given size. 
Besides providing us with a simpler setting in which to reason about
complexity, the absence of values in the upper bounds
is also in line with typical algorithm analysis. 
A long-term goal is to incorporate well-established analysis 
techniques into the formal system of reasoning that we present here.

In this paper we consider a higher-order language, defined in
Section~\ref{sec:target_lang}, over integers and integer lists
with structural list recursion. Since the
interesting analyses are done in terms of the sizes of the lists involved,
we declare all integers to be of some constant size, which we
denote by~$1$.
As an
example of our analysis, let us consider list insertion, defined by
\begin{lstlisting}[basicstyle=\small]
ins (x, nil)   = [x]
ins (x, y::ys) = if (x<=y) then x::y::ys else y::(ins(x, ys))
\end{lstlisting}
A simple cost analysis of 
\lstinline!ins! yields a recursive cost function \lstinline!ins_c! for which
the recurrence argument is the length of the list argument
of \lstinline!ins!:
\[
\lstinline!ins_c!(0) = c_0
\qquad
\lstinline!ins_c!(n+1) = c_1 + \bigl(c_2 \bmax (c_3+\lstinline!ins_c!(n))\bigr)
\]
where the $c_i$ are constants.  The maximum in the
recursion clause ensures that we need not consider the value of the
test \lstinline!x <= y!.  An easily-formalized proof by induction
tells us that $\lstinline!ins_c!(n)\in O(n)$.
But this is not enough for our purposes. 
We want our analyses to be compositional, so that we can directly
use the analysis of, say, \lstinline!ins! in the analysis of
insertion-sort defined by
\begin{lstlisting}
ins_sort xs = fold ins xs nil
\end{lstlisting}
An implicit part of the analysis of \lstinline!ins_sort! uses
the \emph{size} of \lstinline!ins xs! in terms of the size of \lstinline!xs!.
Thus in addition to the cost analysis above, our approach also generates
a size analysis, which we refer to as potential (more on this
terminology momentarily).
The following is a possible potential analysis of \lstinline!ins!:
\[
\lstinline!ins_p!(0) = 1
\qquad
\lstinline!ins_p!(n+1) = (n+1) \bmax (1+\lstinline!ins_p!(n))
\]
where again we take a maximum so that we need not consider the value
of the test \lstinline!x <= y!.

Since we are interested
in higher-order languages, we also need to consider algorithms
which take functions as input. For this, we need some way
to represent the cost associated with using a function; this,
combined with size, is our
notion of \emph{potential}, which we describe more fully
in Section~\ref{sec:complexity_lang}.
For type-level~$0$
objects, we can think of potential as ordinary size, which is why
\lstinline!ins_p! is a function from integers to integers.
The potential of a function essentially encompasses
its cost analysis---the potential of $f$
is a map from potentials $p$ (representing the potential of the
argument) to the complexity of
applying $f$ to an argument of size $p$. 
Our complexity analysis of an expression yields 
a pair consisting of a cost (of computing that expression to obtain
a value) and a potential (representing the size of that value). 
We refer to such a pair as a \emph{complexity}.

The complexity language that we define in Section~\ref{sec:complexity_lang}
allows us to express recurrences like \lstinline!ins_c! and \lstinline!ins_p!.
In Section~\ref{sec:translation} 
we define a translation function $\trans\cdot$ from the target language
to the complexity language so that $\trans{\lstinline!ins!}$ 
encompasses these two recurrences.
The recurrences that result from the translation
will be more complicated than those shown here
because they take complexities, as opposed to sizes,
as arguments. But as shown in Section~\ref{sec:ins_sort},
it is easy to extract 
\lstinline!ins_c! and \lstinline!ins_p! from $\trans{\lstinline!ins!}$.

As another example, consider the \lstinline!map! function defined
by
\begin{lstlisting}
map h nil       = nil
map h (x :: xs) = (h x) :: (map h xs)
\end{lstlisting}
The natural cost analysis of \lstinline!map! yields a recurrence
that depends on the cost of applying \lstinline!h! to account for
the term \lstinline!h x! in the recursion clause.  Note that this
is \emph{not} the cost of \lstinline!h! itself (or any specific
function argument).  Indeed, since any specific function argument is
likely to be expressed as a $\lambda$-abstraction, the cost of
such an argument would be~$1$.  Instead, we must refer to the
\emph{potential} of the function argument~$h$; applied to the potential of
a list element~$x$ (representing the size of~$x$),
we obtain the complexity of the application $h(x)$.  That complexity
is a pair consisting of the cost of the application and its potential.
Taking into account that our integers have constant size 1,
we end up with a cost analysis that looks something like the following,
where $h$ now represents the \emph{complexity} of the function argument:%
\footnote{As we will see, parameters in complexity functions that
correspond to (list-) recursion arguments represent the potential of
such arguments.  Parameters that correspond to non-recursion arguments
represent the complexity of those arguments.}
\[
\lstinline!map_c!(h, 0) = 0
\qquad
\lstinline!map_c!(h, n+1) = (h_p(1))_c + \lstinline!map_c!(h, n)
\]
where $(\cdot)_c$ and $(\cdot)_p$ extract the cost and potential
components of a complexity, respectively.
The potential analysis is straightforward and does not depend on
the function argument:
\[
\lstinline!map_p!(h, 0) = 0
\qquad
\lstinline!map_p!(h, n+1) = 1 + \lstinline!map_p!(h, n)
\]

It is nice that our translation gives the expected costs in these
examples, but of course we want to know that it is sound, in the
sense that $\trans{t}$ bounds the complexity of~$t$.  We prove such
a Soundness Theorem in Section~\ref{sec:soundness} and state
clearly the corollary that an upper bound on the evaluation of cost~$t$ can be
derived directly from~$\trans{t}$.

The function $\trans\cdot$ is computable, and hence
provides a formal link between the program source code and its
complexity bound.
We have implemented a subset of the target and complexity languages, 
translation, and
proof of the Soundness Theorem in Coq, and we discuss some of the
details of this in Section~\ref{sec:implementation}.
The formalization thus provides a mechanism for providing
\emph{certified} upper bounds on the complexity of target language
expressions.   This is in distinction to a more traditional
ad-hoc analysis; such an analysis, even if formalized in a system
such as Coq, is not tied to the source code in a machine-checkable
manner, and so one still does not have a proof that the code to be
executed satisfies the cost bounds that are asserted.

\section{Target language}
\label{sec:target_lang}

The target language is essentially a variant on G\"odel's System~$T$; its
syntax, typing, and operational semantics
are given in Figures~\ref{fig:target_grammar}--\ref{fig:target_semantics}.
The language provides for higher-order programming over integers, booleans, and
integer lists (the \emph{base types}).  
The latter are defined as a recursive datatype, and
the datatype definition automatically provides for structural recursion.
The work in this paper extends to other comparable recursive datatypes;
we treat the special case here to cut down on notation.%
\footnote{Our language does not support general recursion, and
recursive datatypes in general pose some difficulties; we discuss both
issues in Section~\ref{sec:concl_further_work}.}
Since the language is straightforward, we omit many of the details.
A \emph{term} is a typeable expression $\typingG t \tau$.
The operational semantics defines a big-step call-by-value evaluation
relation that relates \emph{closures} $\cl* {\typingG t \tau}\xi$ 
to \emph{values} $\cl v\theta$ (we usually drop the typing details and
just write $\cl t\xi$).
A closure $\cl t\xi$ consists of a term~$t$ and a
\emph{value environment}~$\xi$ that maps variables to values 
such that $\dom\xi$ contains all of $t$'s free variables.
We write $\emptyenv$ for the empty environment.
A value $\cl v\theta$ consists of a \emph{value expression}
and a value environment.  A value expression is any of the
following:
\begin{itemize}
\item A boolean value $\TT$ or $\FF$ or integer value $n\in\mathbf{Z}$;
\item Any finite sequence $(n_0,\dots,n_{k-1})$ of integers;
\item Any expression of the form $\lambda x.r$.
\end{itemize}
We often write $(n,ns)$ for the sequence $(n,n_0,\dots,n_{k-1})$ when
$ns = (n_0,\dots,n_{k-1})$ and write $()$  for the empty sequence.
Since the grammar does not allow lists of higher-order expressions
and the semantics does not
have side-effects, we can safely drop
the environment at leaves of evaluation derivations that derive values
of base type.  We do so in order to simplify the statements of
the rules for evaluating
lists; it is not necessary in practice.
To further clean up notation, we often write $\TT$ or $ns$ instead
of $\cl\TT\emptyenv$ or $\cl{ns}{\emptyenv}$.

The evaluation of 
$\foldexp{r}{s}{x}{xs}{w}{t}$ when $\eval{\cl r\xi}{\cl{(n,ns)}{}}$
deserves some comment.  
A more natural rule might be
\begin{prooftree}
  \AXC{$\eval{\cl r\xi}{\cl{(n,ns)}{}}$}
  \AXC{$\eval{\cl*{\foldexp{ns}{s}{x}{xs}{w}{t}}{\xi}}{\cl {v_0}{\theta_0}}$}
  \AXC{$\eval{\cl t{\xi_1}}{\cl v\theta}$}
  \insertBetweenHyps{\hskip-.025in}
\TIC{$\eval{\cl*{\foldexp{r}{s}{x}{xs}{w}{t}}{\xi}}{\cl v\theta}$}
\end{prooftree}
However, the \lstinline!fold! ``expression'' in the hypothesis is not
well-formed, because $ns$ is not an expression, it is a value.  There is
an obvious isomorphism between the two that we could employ, but then
evaluating the hypothesis \lstinline!fold! expression would require
re-evaluating the list corresponding to~$ns$ which would add (possibly
non-trivially) to the cost of the evaluation.
Thus we choose a fresh variable~$y$
and bind~$y$ to~$ns$ in the environment; as a result, every future
evaluation of the recursion argument will have cost~$1$.\footnote{If
\lstinline!fold! were defined in terms of a more general 
\lstinline!letrec! constructor, then a standard operational semantics
of \lstinline!letrec! as in
\citep{reynolds:theories} would introduce a comparable fresh variable.}

We define the \emph{cost} of a closure $t\xi$, $\cost(t\xi)$, to be the
size of the evaluation derivation of $t\xi$ (it is straightforward to
prove that derivations are unique).  We charge unit cost for arithmetic
operations and count every inference rule.  We can easily
adapt the system for other notions of cost (e.g., counting 
only creation of cons-cells) by modifying the complexity semantics 
described in Section~\ref{sec:complexity_lang}.

\begin{figure}
\begin{align*}
e \in \Exp &::= X \mid Z \mid \lstinline!true! \mid \lstinline!false! \mid 
                \nilexp \mid e :: e \mid e\mathrel{R} e \mid \\
  & \ifexp{e}{e}{e} \mid \lambda x.e \mid e\,e \\
  & \caseexp{e}{e}{x}{xs}{e} \mid \foldexp{e}{e}{x}{xs}{w}{e}
\end{align*}
\caption{Target language grammar.  $X$ ranges over variable identifiers,
$Z$ over integer constants, and $R$ over numerical binary relations
such as $\leq$ and $>$ and numerical binary operators such as~$+$
and~$\times$.}
\label{fig:target_grammar}
\begin{align*}
\sigma,\tau &::= \lstinline!int! \mid \lstinline!bool! \mid 
    \lstinline!int*! \mid \sigma\arrow\tau \\
\lstinline!int*! &::= \lstinline!nil! \mid \lstinline!$::\;$ of (int,int*)!
\end{align*}
\begin{gather*}
  \AXC{$\typingG{e_0}{\lstinline!int*!}$}
  \AXC{$\typingG{e_1}{\sigma}$}
  \AXC{$\typing[\Gamma,x\oftype\lstinline!int!,xs\oftype\lstinline!int*!]
               {e_2}{\sigma}$}
\TIC{$\typingG{\caseexp{e_0}{e_1}{x}{xs}{e_2}}{\sigma}$}
\DisplayProof
\\
  \AXC{$\typingG{e_0}{\lstinline!int*!}$}
  \AXC{$\typingG{e_1}{\sigma}$}
  \AXC{$\typing[\Gamma,x\oftype\lstinline!int!,xs\oftype\lstinline!int*!,
                w\oftype\sigma]
               {e_2}{\sigma}$}
  \insertBetweenHyps{\hskip-0pt}
\TIC{$\typingG{\foldexp{e_0}{e_1}{x}{xs}{w}{e_2}}{\sigma}$}
\DisplayProof
\end{gather*}
\caption{Target language types and typing.  Rules not shown here are the
expected ones.}
\end{figure}

\begin{figure*}
\begin{gather*}
\ndAXC{$\eval{\cl x\xi}{\xi(x)}$}
\DisplayProof
\qquad
\ndAXC{$\eval{\cl{c}{\xi}}{\cl{\underline c}{\emptyenv}}$}
\DisplayProof
\qquad
\ndAXC{$\eval{\cl*{\lambda x.r}\xi}{\cl*{\lambda x.r}{\xi}}$}
\DisplayProof
\\
  \AXC{$\eval{\cl r\xi}{\cl{n_r}\emptyenv}$}
  \AXC{$\eval{\cl s\xi}{\cl{n_s}\emptyenv}$}
  \AXC{$n_r\mathrel R n_s$}
\RightLabel{($R = \mathord<,\mathord\leq,\mathord=,\dotsc$)}
\ndTIC{$\eval{\cl*{r\mathrel R s}{\xi}}{\cl\TT\emptyenv}$}
\DisplayProof
\\
  \AXC{$\eval{\cl r\xi}{\cl{n_r}\emptyenv}$}
  \AXC{$\eval{\cl s\xi}{\cl{n_s}\emptyenv}$}
  \AXC{$\neg(n_r\mathrel R n_s)$}
\RightLabel{($R = \mathord<,\mathord\leq,\mathord=,\dotsc$)}
\ndTIC{$\eval{\cl*{r\mathrel R s}{\xi}}{\cl\FF\emptyenv}$}
\DisplayProof
\\
  \AXC{$\eval{\cl r\xi}{\cl{n_r}\emptyenv}$}
  \AXC{$\eval{\cl s\xi}{\cl{n_s}\emptyenv}$}
  \AXC{$n = n_r\bullet n_s$}
\RightLabel{($\bullet = \mathord+,\mathord-,\mathord\times,\dotsc$)}
\ndTIC{$\eval{\cl*{r\bullet s}{\xi}}{\cl n\emptyenv}$}
\DisplayProof
\\
\ndAXC{$\eval{\cl{\nilexp}{\xi}}{\cl{()}{\emptyenv}}$}
\DisplayProof
\qquad
  \AXC{$\eval{\cl r\xi}{\cl{n}{\emptyenv}}$}
  \AXC{$\eval{\cl s\xi}{\cl{ns}{\emptyenv}}$}
\BIC{$\eval{\cl*{r::s}{\xi}}{\cl{(n,ns)}{\emptyenv}}$}
\DisplayProof
\\
  \AXC{$\eval{\cl r\xi}{\cl*{\lambda x.r_0}{\theta_0}}$}
  \AXC{$\eval{\cl s\xi}{\cl{v_1}{\theta_1}}$}
  \AXC{$\eval{\cl{r_0}{\extend{\theta_0}{x}{\cl{v_1}{\theta_1}}}}
             {\cl v\theta}$}
\TIC{$\eval{\cl*{r\,s}\xi}{\cl v\theta}$}
\DisplayProof
\\
  \AXC{$\eval{\cl r\xi}{\cl{()}{\emptyenv}}$}
  \AXC{$\eval{\cl s\xi}{\cl v\theta}$}
\BIC{$\eval{\cl*{\caseexp{r}{s}{x}{xs}{t}}{\xi}}{\cl v\theta}$}
\DisplayProof 
\qquad
  \AXC{$\eval{\cl r\xi}{\cl{(n,ns)}{\emptyenv}}$}
  \AXC{$\eval{\cl t{\extend \xi {x,xs}{n,ns}}}{\cl v\theta}$}
\BIC{$\eval{\cl*{\caseexp{r}{s}{x}{xs}{t}}{\xi}}{\cl v\theta}$}
\DisplayProof 
\\
  \AXC{$\eval{\cl r\xi}{\cl{()}{\emptyenv}}$}
  \AXC{$\eval{\cl s\xi}{\cl v\theta}$}
\BIC{$\eval{\cl*{\foldexp{r}{s}{x}{xs}{w}{t}}{\xi}}{\cl v\theta}$}
\DisplayProof 
\\
  \AXC{$\eval{\cl r\xi}{\cl{(n,ns)}{\emptyenv}}$}
  \AXC{$\eval{\cl*{\foldexp{y}{s}{x}{xs}{w}{t}}{\xi_0}}{\cl {v_0}{\theta_0}}$}
  \AXC{$\eval{\cl t{\xi_1}}{\cl v\theta}$}
  \insertBetweenHyps{\hskip-.025in}
\TIC{$\eval{\cl*{\foldexp{r}{s}{x}{xs}{w}{t}}{\xi}}{\cl v\theta}$}
\DisplayProof 
\\
\text{where $y$ is fresh}; \xi_0 = \extend\xi y {ns}; \xi_1 = \extend\xi {x,xs,w}{n,ns,\cl {v_0}{\theta_0}}
\end{gather*}
\caption{Target language operational semantics.  $c$ ranges over integer
and boolean constants and $\underline c$ over the corresponding values.
The third hypothesis of the relation and operation rules has unit cost.}
\label{fig:target_semantics}
\end{figure*}

\section{The complexity language}
\label{sec:complexity_lang}

Our goal is to assign a complexity~$\trans t$ to each target language 
expression~$t$.  We have
three desiderata on~$\trans\cdot$:
\begin{itemize}
\item $\trans t$ must provide an upper bound on the cost of evaluating~$t$.
\item $\trans\cdot$ must be compositional.
\item $\trans t$ must not depend on the value of any subexpression of~$t$.
\end{itemize}
Since closures, not expressions, are evaluated in the target language,
$\trans t$ will also be an expression, the meaning of which is determined
by an environment assigning complexities to its free variables.
To say that $\trans \cdot$ is compositional means that $\trans t$ depends
only on expressions $\trans s$ for subexpressions~$s$ of~$t$.
Because of the third constraint, the information we provide
about the evaluation cost of~$t$ is not as precise as that
of as the systems mentioned in the introduction.
However, this constraint is in line with
almost all the work in practical analysis of algorithms and
thus opens up the possibility of using the considerable body
of tools and tricks of analysis of algorithms in our
verifications.

A complexity measure that considers only cost is insufficient if we
want to be able to handle higher-order expressions like 
\lstinline!map!.  To see why, consider any expression
$\lambda x.r$ such that $\lstinline!map!(\lambda x.r)$ is well-typed.%
\footnote{We phrase this informal discussion in terms of evaluating
closed expressions rather than closures.}
Assuming that $\trans\cdot$ is compositional, if $\trans t$ were to provide
information on just the cost of evaluating~$t$, then since
the cost of $\lambda x.r$ is~$1$, the cost of
$\lstinline!map!(\lambda x.r)$ would be independent of~$r$.  What
we need instead is a complexity measure such that $\trans t$
not only captures
the cost of evaluating~$t$, but also the cost of \emph{using}~$t$.
We call this latter notion \emph{potential}, and a complexity will
be a pair consisting of a cost and a potential.  To gain some
intuition for the full definition, we consider the type-level~$0$
and~$1$ cases.  At type-level~$0$, the potential cost of an expression
is a measure of the size of that expression's value.
Now consider a type-level~$1$ expression~$r$.  The \emph{use} of~$r$
is its application to a type-level~$0$ expression~$s$.  The cost of
such an application is the sum of
(i)~the cost of evaluating~$r$ to a value~$\lambda x.r'$;
(ii)~the cost of evaluating $s$ to a value~$v'$;
(iii)~the cost of evaluating $r'[x\mapsto v']$; and
(iv)~a ``charge'' for the inference.  Since (iii) depends (in part)
on the size of $v'$ (i.e., the potential of~$s$), by compositionality
complexities must capture both cost and potential.
Furthermore, (iii) is defined in terms of the potential of $r$
(i.e., the potential of $\lambda x.r'$).  
Thus the potential of a 
type-level~$1$ expression should be a map from type-level~$0$ potentials to
type-level~$0$ complexities, and in general the potential of an
expression of type $\sigma\arrow\tau$ should be a map from potentials of
type-$\sigma$ expressions to complexities of type-$\tau$ expressions.

We now turn to the formal definitions given in
Figures~\ref{fig:cpxy_grammar}--\ref{fig:cpxy_semantics}.  
We define a \emph{complexity
language} over a simple type structure with products
into which our translation function~$\trans\cdot$ will map.
The \emph{full complexity types} consist of the simple types with
products over $\N$, which is intended to be the natural numbers and represents
both costs and potentials of base-type values.
The \emph{complexity} and \emph{potential} types are defined by the
following mutual induction, following our preceding discussion about
the potential of higher-order values:
\begin{enumerate}
\item If $\gamma$ is a potential type, then $\N\cross\gamma$ is a complexity
type.
\item $\N$ is a potential type.
\item If $\gamma$ is a potential type and $\tau$ a complexity type,
then $\gamma\arrow\tau$ is a potential type.
\end{enumerate}
We introduce two notations for potential types~$\gamma$:  
$\pcpxy\gamma = \N\cross\gamma$
and $\parrow\gamma\tau = \pcpxy{(\gamma\arrow\tau)}$ (remember that
the potential of a function is a map from argument potentials to
result complexities).  

The complexity expression constructors roughly mirror the target expression
constructors, and the meaning of the former is intended to capture
the complexity (cost and potential) of the latter.  
The grammar of complexity expressions is as expected for the product-related
types, except we write $t_c$ and $t_p$ for the first and second projections
(for ``cost'' and ``potential''), respectively.  The usual abstraction
and application are replaced by $\lambda_* x.r$ and $r*s$.  
As we will see,
we need an operation that ``applies'' $\trans r$ to $\trans s$; since
both are complexities, and hence pairs, ordinary application does
not suffice.  Thus we introduce the $*$ operator so that we can
define $\trans{rs} = \trans r*\trans s$.  $\lambda_* x.r$ is the corresponding
abstraction operator.  
Conditional expressions are eliminated in the complexity language.
A conditional expression is translated to (essentially) the
maximum of the complexity of its two branches.
This matches our interest in upper bounds on complexity,
ensuring that the complexity of a conditional bounds the cost
of any possible evaluation of that conditional on inputs of a
given size or smaller.
The \lstinline!fold! constructor for recursive
datatypes has a counterpart \lstinline!pfold!
in the complexity language so that
recursive definitions in the target language are translated to
recurrences in the complexity language.
Instead of branching on a constructor, \lstinline!pfold! branches potentials.
Since \lstinline!case! can be seen as a trivial version of \lstinline!fold!,
it has a corresponding counterpart \lstinline!pcase!.%
\footnote{In fact, we could also define conditionals in the target
language to be syntactic sugar for
a \lstinline!case! over the trivial recursive datatype of booleans, and
we would end up with the same complexity semantics.}

Meanings are assigned to complexity terms 
(typeable complexity expressions) through a
denotational semantics.
We take the standard denotation $\den\cdot{}$ of full complexity
types (interpreting $\N$ as the natural numbers).  If $\Gamma$ is
a full complexity type environment, we say that $\xi$ is
\emph{$\Gamma$-consistent} if $\xi(x)\in\den{\Gamma(x)}{}$ for all
$x\in\dom\Gamma$.
Finally, we define a function
$\den\cdot{\mathord-}$ that maps a complexity term
$\typingG e\tau$ and $\Gamma$-consistent environment~$\xi$ to
$\den{\typingG e \tau}\xi\in\den\tau{}$.
We write $\den e\xi$ when $\Gamma$ is clear from context.

The denotational semantics of the complexity expression constructors
describes the cost and potential of the corresponding target-language
constructors.
For example, consider the evaluation of $t = rs$ (again, we phrase
this discussion in terms of closed expressions for clarity):
\[
  \AXC{$\eval{r}{\lambda x.r'}$}
  \AXC{$\eval{s}{v'}$}
  \AXC{$\eval{\extend{r'}{x}{v'}}{v}$}
\TIC{$\eval{rs}{v}$}
\DisplayProof.
\]
$\cost(r)$ and $\cost(s)$ both contribute to $\cost(t)$.  Recalling our
earlier discussion of higher-type potentials,  if $\trans r$ and $\trans s$
are the complexities of $r$ and $s$, then
${\trans r}_p({\trans s}_p)$ (a complexity) gives both the cost of
evaluating $\extend{r'}{x}{v'}$ as well as its potential; but its
potential is also the potential of~$t$.
Thus we define the meaning of~$*$ expressions so that
\[
\den{\trans r*\trans s}{} = 
\bigl(1+\den {{\trans r}_c}{}+\den {{\trans s}_c}{} + \bigl(\den {{\trans r}_p}{}(\den {{\trans s}_p}{})\bigr)_c,
\bigl(\den {{\trans r}_p}{}(\den {{\trans s}_p}{})\bigr)_p\bigr).
\]
We frequently need to ``add cost to a complexity,'' so we
define $\dally(n, (c, p)) = (n+c, p)$. Now we can write, for example,
\[
\den{\trans r*\trans s}{} = 
\dally\bigl(1+\den {{\trans r}_c}{}+\den {{\trans s}_c}{}, \den {{\trans r}_p}{}(\den {{\trans s}_p}{})\bigr).
\]

\lstinline!pcase! and \lstinline!pfold! expressions have a slightly
more complex semantics that might be expected.
Focusing on the former,
the non-zero branch of a \lstinline!pcase! expression must take the
maximum of the two branches, rather than just the second branch, even
though a non-zero potential ought to correspond to a non-empty list.
Because our goal is to establish upper bounds
on complexity (hence potential), it 
may be that the branching expression in the target expression evaluates
to~$\nilexp$, but its translation has a non-zero potential.  As an
example, consider the term~$t$ defined by
\[
\lstinline!case (if true then nil else [0]) of ([0,0], [x,xs]nil)!
\]
The test expression translates to a complexity with potential~$1$;
if $\den{\trans t}{}$ were to take into account only the non-nil branch,
we would conclude that
$t$ has complexity $(7, 0)$ (cost~$7$, potential~$0$), whereas in fact
it has cost~$9$ and size~$2$.
A similar issue arises with \lstinline!pfold!, although in this case the cost is
not an issue; because \lstinline!pfold! is a structural recursion, the base
case expression will be evaluated, and hence its cost included in
the total cost (see Lemma~\ref{lem:cost_pfold_rec} for a precise
statement).

\begin{figure}
\begin{align*}
e ::= &X \mid N \mid e + e \mid e \bmax e \mid (e, e) \mid e_c \mid e_p \mid
       \lambda_* x.e \mid e*e \\
& \pcaseexp{e}{e}{p}{ps}{e} \mid \pfoldexp{e}{e}{p}{ps}{w}{e}.
\end{align*}
\caption{Complexity language expressions.  $X$ is a set of variables and
$N$ a set of constants for each $n\in\N$.}
\label{fig:cpxy_grammar}

\[
\sigma, \tau ::= \N \cross \gamma \qquad
\gamma ::= \N \mid \gamma\arrow\tau
\]
\caption{Complexity and potential types.  The \emph{full} complexity types
consist of the simple types with products over~$\N$.}
\label{fig:cpxy_types}

\begin{gather*}
\ndAXC{$\typing[\Gamma,x\oftype\sigma]{x}{\sigma}$}
\DisplayProof
\quad
\ndAXC{$\typingG n {\N}$}
\DisplayProof
\\
  \AXC{$\typingG r {\N}$}
  \AXC{$\typingG s {\N}$}
\BIC{$\typingG{r+s}{\N}$}
\DisplayProof
\quad
  \AXC{$\typingG r \tau$}
  \AXC{$\typingG s \tau$}
\BIC{$\typingG{r\bmax s}{\tau}$}
\DisplayProof
\\
  \AXC{$\typingG r \N$}
  \AXC{$\typing s \gamma$}
\BIC{$\typingG {(r, s)} {\pcpxy\gamma}$}
\DisplayProof
\quad
  \AXC{$\typingG r {\pcpxy\gamma}$}
\UIC{$\typingG {r_c}{\N}$}
\DisplayProof
\quad
  \AXC{$\typingG r {\pcpxy\gamma}$}
\UIC{$\typingG {r_p}{\gamma}$}
\DisplayProof
\\
  \AXC{$\typing[\Gamma,x\oftype\pcpxy\gamma]{r}{\pcpxy\eta}$}
\UIC{$\typingG {\lambda_* x.r}{\parrow\gamma{\pcpxy\eta}}$}
\DisplayProof
\quad
  \AXC{$\typingG r {\parrow\gamma{\pcpxy\eta}}$}
  \AXC{$\typingG s {\pcpxy\gamma}$}
\BIC{$\typingG {r*s} {\pcpxy\eta}$}
\DisplayProof
\\
  \AXC{$\typingG r \N$}
  \AXC{$\typingG s {\pcpxy\gamma}$}
  \AXC{$\typing[\Gamma,p\oftype\N, ps\oftype\N] t {\pcpxy\gamma}$}
\TIC{$\typingG {\pcaseexp{r}{s}{p}{ps}{t}} {\pcpxy\gamma}$}
\DisplayProof
\\
  \AXC{$\typingG r \N$}
  \AXC{$\typingG s {\pcpxy\gamma}$}
  \AXC{$\typing[\Gamma,p\oftype\N, ps\oftype\N, w\oftype\pcpxy\gamma] t {\pcpxy\gamma}$}
\TIC{$\typingG {\pfoldexp{r}{s}{p}{ps}{w}{t}} {\pcpxy\gamma}$}
\DisplayProof
\end{gather*}
\caption{Complexity language typing rules.  Recall that
$\pcpxy\gamma = \N\cross\gamma$ and
$\parrow\gamma\tau = \pcpxy{(\gamma\arrow\tau)}$.}
\label{fig:cpxy_typing}
\end{figure}

\begin{figure*}
\begin{align*}
\den{\typing[\Gamma,x\oftype\sigma]{x}{\sigma}}\xi &= \xi(x) \\
\den{\typingG n \N}\xi &= n \\
\den{\typingG{r+s}{\N}}\xi &= \den r\xi + \den s\xi \\
\den{\typingG{r\bmax s}{\tau}}\xi &= \den r\xi \bmax  \den s\xi \\
\den{\typingG{(r, s)}{\pcpxy\gamma}}\xi &= (\den r\xi,  \den s\xi) \\
\den{\typingG{r_c}{\N}}\xi &= \pi_0(\den r\xi) \\
\den{\typingG{r_p}{\gamma}}\xi &= \pi_1(\den r\xi) \\
\den{\typingG{\lambda_*x.r}{\parrow\gamma{\pcpxy\eta}}}\xi &=
(1, \llambda p.\den{\typing[\Gamma,x\oftype\pcpxy\gamma]{r}{\pcpxy\eta}}{\extend\xi x {(1, p)}}) \\
\den{\typingG{r*s}{\pcpxy\eta}}{\xi} &= 
\dally(1+\den{r_c}\xi+\den{s_c}\xi, \den{r_p}\xi(\den{s_p}\xi)) \\
\den{\typingG{\pcaseexp{r}{s}{p}{ps}{t}}{\pcpxy\gamma}}\xi &=
\begin{cases}
\den s\xi,&\den r\xi = 0 \\
\den s\xi \bmax \den t{\xi_1},&\den r\xi=q+1
\end{cases}
\\
& \text{where } \xi_1 = \extend\xi{p,ps}{1,q} \\
\den{\typingG{\pfoldexp{r}{s}{p}{ps}{w}{t}}{\pcpxy\gamma}}\xi &=
\begin{cases}
\den s\xi,&\den r\xi = 0 \\
\bigl(2+\den{\rec_c}{\xi_0} + \den{t_c}{\xi_1}, \\
\quad\den{s_p}\xi\bmax\den{t_p}{\xi_1}\bigr),&\den r\xi=q+1
\end{cases}
\\
& \text{where } \rec = \pfoldexp y s p {ps} w t; \\
& \xi_0 = \extend \xi y q; \xi_1 = \extend \xi {p, ps, w} {1, q, (1, \rec_p)}
\end{align*}
\caption{Denotational semantics of complexity terms.  $\llambda p.\dotsb$ 
denotes the semantic function that maps~$p$ to~$\dotsb$.}
\label{fig:cpxy_semantics}
\end{figure*}

\section{Translation from target to complexity language}
\label{sec:translation}

The translation from target language to complexity language is given in
Figure~\ref{fig:expr_trans}.
We assume a bijection between target and complexity variables, so that
when we write $\trans x = x$, the occurrence of $x$ on the left-hand side
is a target variable and the occurrence of $x$ on the right-hand side is
the corresponding complexity variable.
The translation of the simplest target expression constructors gives
expressions that describe directly the cost and size of the corresponding
target expressions; more complex constructors are translated to the
corresponding complexity constructors, and we leave it to the
denotational semantics to extract the complexities from there.  
There is a choice to be made regarding which constructors in the
target language have corresponding constructors in the complexity
language.  Certainly \lstinline!fold! must be in this list, so that
recursive programs are mapped to recurrences.  But whether abstraction
and application should have counterparts, or whether they should be
translated into complexity expressions that directly describe the
denotational semantics of $\lambda_*$ and $*$, is not completely clear.
The choice we have made seems to allow simpler reasoning about the
translated expressions, reasoning that we should be able to easily
formalize in Coq, and which would also be more familiar to programmers.

\begin{figure*}
\[
\typot\baseb = \N
\qquad
\typot{\sigma\arrow\tau} = \typot\sigma\arrow\cpxy\tau
\qquad
\cpxy\tau = \N\cross\typot\tau = \pcpxy{\typot\tau}
\]
\begin{align*}
\trans x &= x \\
\trans c &= (1, 1) \\
\trans{r\mathbin{R}s} &= (2+\trans r_c+\trans s_c, 1) \\
\trans{\nilexp} &= (1, 0) \\
\trans{r::s} &= (1+\trans r_c+\trans s_c, 1+\trans s_p) \\
\trans{\ifexp{r}{s}{t}} &= \dally(1+\trans r_c,\trans s\bmax\trans t) \\
\trans{\lambda x.r} &= \lambda_*x.\trans r \\
\trans{rs} &= \trans r*\trans s \\
\trans{\caseexp{r}{s}{x}{xs}{t}} &= \dally(1+\trans{r}_c, \\
  &\qquad \lstinline!pcase!~\trans r_p~\lstinline!of!~(\trans s,
  [p, ps]\extend{\trans t}{x, xs}{(1, p),(1,ps)})) \\
\trans{\foldexp{r}{s}{x}{xs}{w}{t}} &= \dally(1+\trans{r}_c, \\
  &\qquad\lstinline!pfold!~\trans t_p~\lstinline!of!~(\trans s,
  [p, ps, w]\extend{\trans t}{x,xs}{(1,p),(1,ps)}))
\end{align*}
\caption{Translation from target types and expressions 
to complexity types and expressions.
$c$~ranges over integer and boolean constants.}
\label{fig:expr_trans}
\end{figure*}

Before proceeding to examples, we note that translation preserves
type derivations.
For a target type context~$\Gamma$, define
$\trans\Gamma = \setst{(x \oftype \trans\tau)}{(x\oftype\tau)\in\Gamma}$.
\begin{prop}
If $\typingG r\tau$, then $\typing[\trans\Gamma]{\trans r}{\trans\tau}$.
\end{prop}

\subsection{Examples}
\label{sec:examples}

We show the results of translating the insertion-sort and map functions
in this section.  We freely transform expressions in the complexity
language according to validities in the semantics, for example adding
natural numbers and computing maximums when possible.  The soundness
of each such equality or inequality is easily proved, and so such
transformations could easily be incorporated in a formalized proof
that simplifies $\trans{t}$ into a more amenable form.

\subsubsection{Insertion-sort}
\label{sec:ins_sort}
{
\def\ins{\mathit{ins}}
\def\isort{\mathit{isrt}}
\def\map{\mathit{map}}
We start by considering the list-insertion function
defined by
\begin{lstlisting}[basicstyle=\small]
ins = $\lambda$x.$\lambda$xs.fold xs of (x :: nil, 
      [y, ys, w] if x <= y then x :: y :: ys else y :: w)
\end{lstlisting}
and showing that it has a linear running time in the
size of its input.
Translating directly yields
\begin{multline*}
\trans{\lstinline!ins!} = \lambda_*x,xs.\lstinline!pfold $\;xs_p\;$ of ($(2+x_c, 1)$,! \\
[p, ps, w]\dally(4+x_c, (4+x_c, 2+ps)\bmax(2+w_c,1+w_p))).
\end{multline*}
If we write $f_{\ins}(x, xs)$ for $\trans{\lstinline!ins!}*x*xs$,
apply the equations for the denotations of $\lambda_*$ and $*$, and use
properties of the order on natural numbers, we have
\begin{align*}
f_{\ins}(x, xs)
  &= \dally(4+x_c+xs_c, \lstinline!pfold $\;xs_p\;$ of ($(3, 1)$,! \\
  &\qquad [p, ps, w] \dally(5, (5,2+ps)\bmax(2+w_c,1+w_p)))) \\
  &\leq \dally(4+x_c+xs_c, \lstinline!pfold $\;xs_p\;$ of ($(3, 1)$,! \\
  &\qquad [p, ps, w] (10+w_c, (2+ps)\bmax(1+w_p)))
\end{align*}
To turn this into a recognizable form,
set 
\[
g(z) = \lstinline!pfold!~z~\lstinline!of!~\bigr((3, 1),
  [p, ps, w](10+w_c, (2+ps)\bmax(1+w_p))\bigl)
\]
so that $f_{\ins}(x, xs) = \dally(4+x_c+xs_c,g(xs_p))$.
Rewriting $g$ as a pair of recurrences (one each for cost and potential)
we obtain
\begin{align*}
g_c(0) &= 3 & g_c(q+1) &= 13+g_c(q) \\
g_p(0) &= 1 & g_p(q+1) &= (2+q)\bmax(1+g_p(q))
\end{align*}
Here again we have used the equations from the denotational
semantics of \lstinline!pfold!.  For example, the expression
for $g_c(q+1)$ is $2+\rec_c+t_c$ where $\rec$ is the 
recursive call (i.e., $g(q)$) and 
\[
t = (10+w_c,(2+ps)\bmax(1+w_p))[w\mapsto (1, \rec_p)] =
(11, (2+ps)\bmax(1+\rec_p)).
\]
A straightforward induction establishes
$g(z) \leq (13z+3, z+1)$ and hence
$f_{\ins}(x, xs) \leq (13xs_p+7+x_c+xs_c, xs_p+1)$.

Continuing, we now consider the insertion-sort function
defined by
\begin{lstlisting}[basicstyle=\small]
ins_sort = $\lambda$xs.fold xs of (nil, [y, ys, w](insert y w))
\end{lstlisting}
Following the approach above and writing
$f_{\isort}(xs)$ for $\trans{\lstinline!ins_sort!}*xs$, we have
\begin{gather*}
g(z) = \pfoldexp z {(1, 0)} p {ps} w {f_{\isort}((1, p), w)} \\
f_{\isort}(xs) = \dally(2+xs_c, g(xs_p)) \\
\begin{align*}
g_c(0) &= 1 & g_c(q+1) &= 2+g_c(q) + \bigl(f_{\ins}\bigl((1, 1), (1, g_p(q))\bigr)\bigr)_c \\
g_p(0) &= 0 & g_p(q+1) &= \bigl(f_{\ins}\bigl((1, 1), (1, g_p(q))\bigr)\bigr)_p
\end{align*}
\end{gather*}
Using our bound on $f_{\ins}$, a proof by induction to show
that $g_p(q)\leq q$, and then this last inequality to simplify
the bound on $g_c(q+1)$ we have
\begin{align*}
g_c(0) &\leq 1 & g_c(q+1) &\leq 11+g_c(q)+13q \\
g_p(0) &\leq 0 & g_p(q+1) &\leq 1+g_p(q)
\end{align*}
from which we conclude
$g(z)\leq (13z^2+11z+1,z)$ and hence
$f_{\isort}(xs) \leq (13xs_p^2 + 9xs_p + 3 + xs_c, xs_p)$.

Although we don't often think of the analysis of insertion-sort
involving an analysis
of size (something we are forced to do when working with complexities),
in fact such size analyses are usually implicit.  For example, in the
standard analysis of insertion-sort, we usually implicitly make use
of the correctness of the algorithm to assert that the length of
$\lstinline!ins_sort!(xs)$ is the same as the length of~$xs$.

It is also worth noting that the analysis of $\trans{\lstinline!ins_sort!}$
only depends on properties of $\trans{\lstinline!insert!}$, namely
that $f_{\ins}(x, xs)\leq (a\cdot xs_p+x_c+xs_c+b, xs_p+1)$.
One way to see this is to note that we could have defined
a general fold function
\begin{lstlisting}[basicstyle=\small]
list_fold = $\lambda$f, xs, a.fold xs of (a, [y, ys, w](f y w))
\end{lstlisting}
and analyzed $\trans{\lstinline!list_fold!}$.  The analysis would be
in terms of $f_c$ and $f_p$.
We could then ``plug in'' assumptions about $f_c$ and $f_p$ to analyze
concrete instances of the general fold function such as insertion-sort.
Such assumptions and analyses could be exact or asymptotic, as the application
demands.
One can certainly imagine that such analyses are likely to
be important when reasoning about large modular programs.

\subsubsection{Map}
\label{sec:map}

Perhaps surprisingly, the analysis of the higher-order map function
defined by
\begin{lstlisting}[basicstyle=\small]
map = $\lambda$h.$\lambda$xs.fold xs of (nil, [y, ys, w](h(y) :: w))
\end{lstlisting}
is more straightforward than for insertion-sort.  Again following the
approach above and writing
$f_\map(h, xs)$ for $\trans{\lstinline!map!}*h*xs$, we have
\begin{gather*}
g(z) = \pfoldexp z {(1, 0)} p {ps} w {(4+(h_p(p))_c, 1+w_p)} \\
f_\map(h, xs) = \dally(4+h_c+xs_c, g(xs_p)) \\
\begin{align*}
g_c(0) &= 1 & g_c(q+1) &= 7+(h_p(1))_c + g_c(q) \\
g_p(0) &= 0 & g_p(q+1) &= 1+g_p(q)
\end{align*}
\end{gather*}
The cost of applying $h$ to any element of $xs$ is fixed, because
the size (potential) of every integer is~$1$; call this cost~$C$.  
We conclude
that $g(z) \leq ((7+C)z+1, z)$ and hence
$f_\map(h, xs) \leq ((7+C)xs_p+5+h_c+xs_c, xs_p)$.
}

\section{Soundness of the translation}
\label{sec:soundness}

We saw in our examples that the translations of the target-language
programs yield expressions that, modulo some manipulations, describe
the expected
upper bounds on the complexities (and hence evaluation costs) of the
original programs.  One might worry that the manipulations themselves
are the source of success, rather than the translation.
Our main goal is to show that (in an appropriate sense),
$\cost(t) \leq \cost(\den{\trans t}{})$ for all programs.
The challenge in doing so is that cost (and potential) is not
compositional:  $\cost(r\,s)$ is not defined solely in terms of the
cost of subexpressions of~$r$ and~$s$.  To get around this,
we define
a relation between target language closures and complexities,
which we then generalize to target and complexity
expressions.  The relation itself is essentially a logical
relation \citep[Ch.~8]{mitchell:foundations}
that allows us to prove our main Soundness result by induction
on terms.  With this in mind,
we start by defining
the \emph{bounding relation}~$\bounded$ as follows.  Let
$t$ be a target-language expression, $\xi$ a value environment defined
on the free variables of~$t$,
$\sigma$ be a target language type, and
$\chi$ be a complexity such that
$\chi\in\den{\trans\sigma}{}$.  We define
$\cl t\xi\bounded_\sigma\chi$ if $\eval{\cl t\xi}{\cl v\theta}$ implies
\begin{itemize}
\item $\cost(\cl t\xi)\leq\cost(\chi)$; and
\item $\cl v\theta\vbounded_\sigma\pot(\chi)$.
\end{itemize}
The \emph{value bounding relation}~$\vbounded$ relates values to potentials:
\begin{itemize}
\item $\cl*{\TT}\theta\vbounded_{\lstinline!bool!} 1$, 
$\cl*{\FF}\theta\vbounded_{\lstinline!bool!} 1$,
$\cl n\theta\vbounded_{\lstinline!int!} 1$.
\item $\cl{(n_0,\dots,n_{k-1})}\theta\vbounded_{\lstinline!int*!} p$ 
if $k\leq p$.
\item $\cl*{\lambda x.r}\theta\vbounded_{\sigma\arrow\tau} p$ 
if whenever $\cl z\eta\vbounded_\sigma q$,
then $\cl r{\extend\theta x{\cl z\eta}}\bounded_\tau p(q)$.
\end{itemize}
We will usually drop the type-subscript from $\bounded$ and $\vbounded$.
If $\xi$ and $\xi^*$ are value- and complexity- environments respectively,
we write $\xi\bounded\xi^*$ to mean that
$\cl x\xi\bounded_\sigma\xi^*(x)$ 
whenever $x\in\dom\xi^*$ and $\xi^*(x)\in\den{\trans \sigma}{}$.
If $t$ is a target expression and $\typing[\Gamma^*]{t^*}{\tau}$ 
a complexity term,
then we write $t\bounded \typing[\Gamma^*]{t^*}{\tau}$ to mean that 
$\cl t\xi\bounded\den{t^*}{\xi^*}$
whenever $\xi^*$ is a $\Gamma^*$-consistent environment and $\xi\bounded\xi^*$.

Our main theorem is the following:

\begin{thm}[Soundness]
If $\typingG t\tau$, then $t\bounded\trans {\typingG t \tau}$.
\end{thm}

\begin{cor}
If $\typingE t\tau$, then $\cost(\cl t\emptyenv)\leq\cost(\trans{\typingE t\tau})$.
\end{cor}

The reader may be concerned by the seeming lack of a connection
between any putative value environment $\xi$ under which $t$ may
be evaluated and $\Gamma$; in particular, there is no connection between
$\xi(x)$ and $\Gamma(x)$ for $x\in\fv(t)$.  
A short response is that the Soundness Theorem is typically applied only to
closed terms (as in the Corollary) 
and hence there is no concern, and the reader is
free to accept this and ignore the remainder of this paragraph.
A longer response is that this observation has to do with the
fact that our evaluation semantics places no restrictions upon
the value environment to enforce reasonable ``typing.''  
Indeed, no restrictions can be placed, as we have not defined a notion
of typing for values, although defining such a notion is not
technically difficult.
However, in applying the Soundness Theorem and chasing the definitions,
one is forced to apply it to value environments that assign ``reasonable''
values to the variables.  For suppose that 
$\typingG t\tau$; $t\bounded\trans{\typingG t \tau}$;
$\xi$ is a value environment such that
$\eval{\cl t\xi}{\cl v\theta}$; and that $x\in\fv(t)$.  
Since $x\in\fv(t)$, $x\in\dom\Gamma$.
Take
any $\trans\Gamma$-consistent environment~$\xi^*$ such that
$\xi\bounded\xi^*$.  Suppose that $\xi^*(x)\in\den{\trans\sigma}{}$;
by the definition of consistency, $\trans\Gamma(x)=\trans\sigma$ and
hence $\Gamma(x)=\sigma$.  But in addition, since
$\cl x\xi\bounded_\sigma\xi^*(x)$ and $\eval{\cl x\xi}{\xi(x)}$,
$\xi(x)\vbounded_\sigma\xi^*(x)$.  Examining the definition of
$\vbounded_\sigma$, we see that this ensures that
the ``type'' of $\xi(x)$ (had we defined it) must be
$\sigma$.

Before proving the Soundness Theorem, we establish some preliminary
lemmas.  The first three are consequences of definitions
and syntactic manipulation.

{
\def\size#1{\cost(#1)}
\def\clos#1#2{\cl{#1}{#2}}
\let\ble\bounded
\let\Empty\emptyenv
\def\judge#1#2#3#4{\eval{\cl{#1}{#2}}{\cl{#3}{#4}}}
\def\pjudge#1#2#3#4{\eval{\cl*{#1}{#2}}{\cl{#3}{#4}}}
\def\bpjudge#1#2#3#4{\cl*{#1}{#2}\bounded{#4}}
\def\bjudge#1#2#3#4{\cl{#1}{#2}\bounded{#4}}
\def\pbjudge#1#2#3#4{\cl{#1}{#2}\vbounded{#4}}
\def\bigdally#1#2{\dally(#1,#2)}

\begin{lemma}
\label{lem:boundmax}
If $\cl t {\xi} \bounded_{\tau} \chi$ then for all 
$\chi' \in \den{\trans{\tau}}{}$, $\bjudge t {\xi} {\tau} {\chi \vee \chi'}$.
\end{lemma}

\begin{lemma}
\label{lem:env_ext}
If $\xi\bounded\xi^*$ and $\cl v\theta\vbounded q$, then
$\extend\xi x {\cl v\theta}\bounded\extend{\xi^*}{x}{(1, q)}$.
\end{lemma}

\begin{lemma}
\label{lem:substitution}
$\den t{\extend\xi{x}{(a,b)}} = \den{\extend t x{(a,y)}]}{\extend\xi y b}$
where $y$ is a fresh variable.
\end{lemma}

The next two lemmas establish bounds related to recursive definitions.

\begin{lemma}
\label{lem:cost_pfold_rec}
For all complexity expressions $r$,
\[
\size{\den{s}\xi^{\ast}} \le \cost(\den{\pfoldexp r s p {ps} w t}\xi^{\ast}).
\]
\end{lemma}
\begin{proof}
Formally we prove by induction on~$q$ that for all complexity
expressions~$r$ and environments~$\xi^*$, if 
$\den r\xi^* = q$, then
\[
\size{\den{s}\xi^{\ast}} \le \cost(\den{\pfoldexp r s p {ps} w t}\xi^{\ast}).
\]
If $q=0$, then the two sides of the inequality are in fact equal.
Suppose $\den r\xi^* = q+1$.  Then
\[
\cost(\den{\pfoldexp r s p {ps} w t}\xi^{\ast}) =
2+\cost(\rec) + \cost(\den t{\extend {\xi^*}{p,ps,w}{1,q,(1,\pot(\rec))}})
\]
where 
$\rec = \den{\pfoldexp y s p {ps} w t}{\extend{\xi^*}{y}q}$.  By the
induction hypothesis, $\cost(\den s\xi^*)\leq\cost(\rec)$, and so the
claim follows.
\end{proof}

\begin{lemma}
\label{lem:fold_rec}
Suppose $s \ble \trans s$ and $t \ble \trans t$. Fix $\xi \ble \xi^{\ast}$
and let $\xi_0 = \extend\xi y {(n_0,\ldots,n_{k-1})}$ and
$\xi_0^{\ast} = \extend{\xi^{\ast}}y q$, where $k \le q$.
Assume $y$ is not free in either $s$ or $t$. Then 
\[
\cl{\foldexp y s x {xs} w t}{\xi_0} \bounded
    \bigdally{2}{\den{\pfoldexp y {\trans s} p {ps} w {t^\prime}}\xi_0^{\ast}}
\]
where $t' = \extend{\trans t}{x,xs}{(1,p), (1,ps)}$.
\end{lemma}
\begin{proof}
We prove the claim by induction on $k$. If $k=0$ then the cost of the
fold expression is $2 + \size{\clos s {\xi_0}}$, and the value of the fold 
expression under $\xi_0$ is the value to which $s \xi_0$ evaluates. 
The cost bound is proved by Lemma \ref{lem:cost_pfold_rec}:
\[
\size{\clos{\foldexp y s x {xs} w t}{\xi_0}} \le 2 + (\trans s\xi_0^{\ast})_c \le
2 + (\den{\pfoldexp y {\trans s} p {ps} w {t'}}\xi_0^{\ast})_c
\]
By the IH
$\pbjudge v {\theta} {\tau} {(\trans s\xi_0^{\ast})_p}$.
The potential of $\bigdally 2 {\den{\pfoldexp y {\trans s} p {ps} w {t'}}\xi_0^{\ast}}$
is either $(\trans s\xi_0^{\ast})_p$ or $(\trans s\xi_0^{\ast})_p \vee (\den{t'}\xi_0^{\ast\ast})_p$
for some environment $\xi_0^{\ast\ast}$, we know from Lemma \ref{lem:boundmax} that
the potential bound holds.

Suppose $\xi_0(y) = (n,ns)$ with $ns = (n_0,\ldots,n_{k-1})$. Then 
$\xi_0^{\ast}(y) = q+1$ for some $q \ge k$. 
By the induction hypothesis we know that
\[
\cl{\foldexp y s x {xs} w t} {\extend\xi y {ns}}\bounded
\bigdally 2 \rec
\]
where $\rec = \den{\pfoldexp y {\trans s} p {ps} w {t'}}{\extend{\xi^{\ast}}y q}.$
So if 
\[\pjudge {\foldexp y s x {xs} w t} {\extend\xi y {ns}} {v'} {\theta'}\]
then $\pbjudge {v'} {\theta'} {\tau} {\rec_p}$, which means
\[
\extend\xi{x,xs,w}{n, ns, v'\theta'} \ble
	 \extend{\xi^{\ast}}{x, xs, w}{(1,1), (1,q), (1,\rec_p)}
\]
Let $\xi_1 = \extend{\xi_0}{x, xs, w}{n, ns, v'\theta'}$
and let $\xi_1^{\ast} = 
\extend{\xi_0^{\ast}}{p, ps, w}{1, q, (1,\rec_p)}$.
Since $y$ does not occur free in $t$ or $t'$, by 
Lemma \ref{lem:substitution}, 
$\bjudge t {\xi_1} {\tau} {\den{t'}\xi_1^{\ast}}$.
So 
\begin{align*}
\cost(\clos{\foldexp y s x {xs} w t}{\xi_0})
&= 2 + \size{\clos{\foldexp y s x {xs} w t}
	 {\extend\xi{y}{ns}}} + \size{\clos t {\xi_1}} \\
&\le 2 + (2 + \rec_c) + \left(\den{t'}\xi_1^{\ast}\right)_c \\
&= 2 + \left(\den{\pfoldexp y {\trans s} p {ps} w {t'}}\xi_0^{\ast}\right)_c
\end{align*}
For the potential bound, if $\pjudge {\foldexp y s x {xs} w t} {\xi_0} v {\theta}$,
then $\judge t {\xi_1} v {\theta}$, so $\pbjudge v {\theta} {\tau} {\pot(\den{t'}\xi_1^{\ast})}$.
By Lemma \ref{lem:boundmax}, we have 
\[
\pbjudge v {\theta} {\tau} {(\trans s\xi_0^{\ast})_p \vee (\den{t'}\xi_1^{\ast})_p}  =
\pot\left(\den{\pfoldexp y {\trans s} p {ps} w {t'}}\xi_0^{\ast}\right).
\]
\end{proof}
}

\begin{proof}[Proof of Soundness Theorem]
By induction on the derivation of $\typingG t\tau$.  Most
of the proof is similar to the analogous proof in
\citep{danner-royer:two-algs}, so we do just a few cases here.
In this proof we will write $\trans s\xi^*$ for
$\den{\trans s}{\xi^*}$.

Suppose $\typingG {(r::s)}{\lstinline!int*!}$
and fix $\xi\bounded\xi^*$.  We have that
$\trans{r::s} = (1+\trans{r}_c+\trans{s}_c, 1+\trans{s}_p)$.
For the cost bound,
\[
\cost(\cl*{r::s}{\xi}) = 1+\cost(\cl r\xi)+\cost(\cl s\xi) \leq
1+\cost(\trans r\xi^*) + \cost(\trans s\xi^*) = \cost(\trans{r::s}\xi^*).
\]
For the potential bound, if $\eval{\cl*{r::s}{\xi}}{(n,n_0,\dots,n_{k-1})}$,
then $\eval {\cl s\xi}{(n_0,\dots,n_{k-1})}$, so by the induction
hypothesis, $\pot(\trans s\xi^*)=q$ for some~$q\geq k$; hence
$\pot(\trans{r::s}\xi^*)= 1+q\geq 1+k$.

Suppose $\typingG {\lambda x.r}{\sigma\arrow\tau}$, with
$\typing[\Gamma,x\mapsto\sigma] r\tau$, and fix $\xi\bounded\xi^*$.
Verifying the cost bound is trivial, so we focus on the potential bound.
Set $p = \trans{\lambda x.r}_p\xi^* =
\llambda v.\trans r\xi^*[x\mapsto (1, v)]$.
We must show that if
$\cl z\eta\vbounded_\sigma q$, then
$\cl r{\extend\xi x{\cl z\eta}}\bounded_\tau p(q) = 
{\trans r}\extend{\xi^*}x{(1, q)}$.
Since $\cl z\eta\vbounded q$, by Lemma~\ref{lem:env_ext}
$\extend\xi x{\cl z\eta}\bounded\extend{\xi^*}x{(1, q)}$, and so the
claim follows by
the induction hypothesis for~$r$.

Suppose $\typingG {r\,s}{\tau}$, with $\typingG r{\sigma\arrow\tau}$
and $\typingG s\sigma$, and fix $\xi\bounded\xi^*$.
By unraveling definitions, 
\[
\trans{rs}\xi^{\ast} = 
\dally(1+(\trans r\xi^{\ast})_c + (\trans s\xi^{\ast})_c,
      {(\trans r\xi^{\ast})_p (\trans s\xi^{\ast})_p}).
\]

Suppose
$\eval {\cl*{rs} {\xi}} {\cl v {\theta}}$ by the following evaluation rule:
\begin{prooftree}
      \AXC{$\eval {\cl r {\xi}} {\cl {\lambda x.r'} {\theta_1}}$}
      \AXC{$\eval {\cl s {\xi}} {\cl w {\theta_2}}$}
      \AXC{$\eval {\cl {r'} {\extend{\theta_1}x {\cl w\theta_2}}} {\cl v {\theta}}$}
\TIC{$\eval {\cl*{r s} {\xi}} {\cl v {\theta}}$}
\end{prooftree}
We first show that
\[
\cl {r'} {\extend{\theta_1}x {\cl w\theta_2}} \bounded
    {(\trans r\xi^{\ast})_p (\trans s\xi^{\ast})_p}.
\tag{$*$}
\]
Since $\eval {\cl  r {\xi}} {\cl* {\lambda x.r'} {\theta_1}}$, we know that
$\cl* {\lambda x.r'} {\theta_1} \vbounded
{\pot\trans r\xi^{\ast}}$. Hence if
$\cl z {\eta} \vbounded p$ then
$\cl {r'} {\extend{\theta_1}x {\cl z\eta}} \bounded
         {(\trans r\xi^{\ast})_p(p)}$.
Since $\cl s\xi\bounded\cl {\trans s}\xi^*$ and 
$\eval{\cl s\xi}{\cl w\theta_2}$, we have that
$\cl w {\theta_2} \vbounded {(\trans s\xi^{\ast})_p}$, and $(*)$ follows.

To establish the cost bound, 
we compute
\begin{align*}
\cost(\cl*{rs}{\xi}) &= 1 + \cost(\cl r {\xi})
+ \cost(\cl s {\xi}) + \cost(\cl {r'} {\theta_1[x \mapsto w\theta_2]}) \\
&\le 1 + \cost(\trans r\xi^{\ast}) + \cost(\trans s\xi^{\ast}) +{} \cost
\left((\trans r\xi^{\ast})_p (\trans s\xi^{\ast})_p)\right) \\
&= \cost(\trans{rs}\xi^{\ast})
\end{align*}
with the inequality following from the induction hypotheses and $(*)$. 
The potential bound follows from $(*)$ and the fact that
$\eval{\cl{r'}{\theta_1[x\mapsto \cl w\theta_2]}}{\cl v\theta}$.

{%
\def\size#1{\cost(#1)}
\def\clos#1#2{\cl{#1}{#2}}
\let\ble\bounded
\let\Empty\emptyenv
\def\judge#1#2#3#4{\eval{\cl{#1}{#2}}{\cl{#3}{#4}}}
\def\pjudge#1#2#3#4{\eval{\cl*{#1}{#2}}{\cl{#3}{#4}}}
\def\bpjudge#1#2#3#4{\cl*{#1}{#2}\bounded{#4}}
\def\bjudge#1#2#3#4{\cl{#1}{#2}\bounded{#4}}
\def\pbjudge#1#2#3#4{\cl{#1}{#2}\vbounded{#4}}
\def\bigdally#1#2{\dally(#1,#2)}
Suppose
$\typingG {\foldexp r s x {xs} w t}{\tau}$ and set
$t' = \extend{\trans t}{x,xs}{(1,p),(1,ps)}$ so that
\[
\trans{\foldexp r s x {xs} w t} =
\pfoldexp {{\trans r}_p} {\trans s} p {ps} w t'.
\]
Suppose $\eval{\cl r\xi}{\cl{ns}{}}$;
we prove the claim by induction on $ns$.
Suppose $ns=(\,)$.  We start by computing
the cost:
\begin{align*}
\cost(\cl{\foldexp r s x {xs} w t}{\xi})
&= 1 + \size{\clos r {\xi}} + \size{\clos s {\xi}} \\ 
&\le 1 + (\trans r \xi^{\ast})_c + (\trans s \xi^{\ast})_c \\
&\le 1 + (\trans r \xi^{\ast})_c
  + (\den{\pfoldexp {\trans r_p} {\trans s} p {ps} w {t^\prime}}{\xi^{\ast}})_c \\
&= \cost(\trans{\foldexp r s x {xs} w t} \xi^{\ast}).
\end{align*}
The second inequality is an equality if ${\trans r}_p = 0$,
or follows from Lemma~\ref{lem:cost_pfold_rec} if ${\trans r}_p > 0$.
Turning to the potential bound,
if 
$\pjudge {\foldexp r s x {xs} w t} {\xi} v {\theta}$ then 
$\judge s {\xi} v {\theta}$ as well and so
$\pbjudge v {\theta} {\tau} {\pot(\trans s \xi^\ast)}$.
If ${\trans r}_p = 0$, this suffices to verify the claim.
If ${\trans r}_p > 0$, use Lemma~\ref{lem:boundmax}.
This finishes
the case in which $\eval{\cl r\xi}{(\,)}$.

Suppose $\eval{\cl r {\xi}} {(n,ns)}$, where
$ns = (n_0,\dots,n_{k-1})$.
By the induction hypothesis for~$r$,
$\trans r_p \xi^{\ast} = {q+1}$ for some $q \ge k$.
Make the following definitions:
\begin{itemize}
\item $\xi_0 = \extend\xi y{{ns}}$; 
	$\xi_0^* = \extend{\xi^*}y q$.
\item $\rec^t = \foldexp y s x {xs} w t$.
\item $\rec = \den{\pfoldexp y {\trans s} p {ps} w {t^\prime}}{\xi_0^{\ast}}$.
\end{itemize}
By Lemma~\ref{lem:fold_rec}, 
$\bjudge {\rec^t} {\xi_0} {\tau} {\dally(2, \rec)}$.
So if $\judge {\rec^t} {\xi_0} {v'} {\theta'}$
then $\pbjudge {v'} {\theta'} {\tau} {\rec_p}$.  Now set
\begin{itemize}
\item $\xi_1 = \extend\xi{x,xs,w}{n, ns, v'\theta'}$.
\item $\xi^{\ast}_1 = \extend{\xi^{\ast}}{p,ps,w}{1, q, (1,\rec_p)}$.
\end{itemize}
By Lemma~\ref{lem:env_ext} $\xi_1\bounded\xi_1^*$,
so by the induction hypothesis for~$t$ and Lemma~\ref{lem:substitution}, 
we have $\bjudge t {\xi_1} {\tau} {\den{t'}\xi^{\ast}_1}$.

For the cost bound when 
$\eval{\cl r \xi}{{(n,ns)}}$, we have
\begin{align*}
\cost(\cl{\foldexp r s x {xs} w t}\xi)
&= 1 + \size{\clos r {\xi}} + \size{\rec^t \extend\xi{y}{ns}}
   + \size{\clos t {\xi_1}} \\
&\le 1 + (\trans r \xi^{\ast})_c + (2 + \rec_c) + (\den{t'}{\xi^{\ast}_1})_c \\
&= 1 + (\trans r\xi^{\ast})_c +
\cost(\den{\pfoldexp {\trans r_p} {\trans s} p {ps} w {t^\prime}}{\xi^{\ast}}) \\
&= \cost(\trans{\foldexp r s x {xs} w t}\xi^{\ast})
\end{align*}
For the potential bound, suppose $\pjudge {\foldexp r s x {xs} w t} {\xi} v {\theta}$.
Then we must have $\judge t {\xi_1} v {\theta}$ as well;
hence $\pbjudge v {\theta} {\tau} {\pot(\den{t'}{\xi^{\ast}_1})}$.
So by~(\ref{lem:boundmax}), 
\[
\pbjudge v {\theta} {\tau} {\pot (\trans s \xi^{\ast}) \vee \pot (\den{t'}{\xi^{\ast}_1})} = {}
\pot(\trans{\foldexp r s x {xs} w t}\xi^{\ast}).
\]
}%
\end{proof}

\section{Implementation in Coq}
\label{sec:implementation}

We have implemented a subset of the target and complexity languages in Coq
that includes the simply typed $\lambda$-calculus with integer and boolean
operations. For this subset we have implemented the translation function
and proof of the Soundness Theorem.
The current development may be found at 
\begin{center}
\url{http://wesscholar.wesleyan.edu/compfacpub}.  
\end{center}
When complete,
the formalization will provide a mechanism for certified
upper bounds on the cost of programs in the target language.
Since the denotational semantics is built on Coq's built-in
type system, one can use the formalized Soundness Theorem
to establish a bound on a given target program, then continue
to reason in Coq to simplify the bound, for example
establishing a closed form if desired.

Although we use Coq to formalize the translation, our system 
is really a blend of external and internal verification.
Our long-term goal is 
to be able to translate from programs written
in a language such as SML or OCaml.  Such programs
would not have cost information directly associated with them via
annotations (such as done by
\citet{danielsson:popl08}) or a type system.
In this sense, our approach follows that of~CFML 
\citep{chargueraud:icfp10}, in which
Caml source code is translated into a formula that can be used
to verify post-conditions.

The main non-trivial aspect of the development is the definition
of the bounding relation.
The bounding relation is
a simultaneous recursive definition to two relations,
$\bounded$ and $\vbounded$.  In Coq the difficulty arises in defining
$\bounded_{\sigma\arrow\tau}$ in
terms of $\vbounded_{\sigma\arrow\tau}$ (in turn defined
in terms of $\bounded_\tau$), which is not a structural descent on type.  
We resolve this by defining
subsidiary versions of $\bounded$ and $\vbounded$ that take a
natural number argument, and which are structurally decreasing on
that argument.  $\bounded$ and $\vbounded$ are then defined in
terms of these relations, using a numeric value sufficiently high
that the inductive definition is guaranteed to terminate by reaching
a base type, rather than a value of~$0$ for the numeric argument.%
\footnote{We also implemented a version in which $\vbounded$ is
inlined into the definition of $\bounded$ as an anonymous fixpoint,
but found that it became even more tedious to prove the required
lemmas, because we often have to reason specifically about $\vbounded$.}
With this out of the way, the proof of the Soundness Theorem proceeds
more-or-less as described in this paper.

\section{Related work}
\label{sec:related_work}

The core idea of this paper is not new.  
There is a reasonably extensive literature over the last several
decades on (semi-)automatically constructing resource bounds
from source code.
The first work naturally concerns itself with first-order
programs.
\citet{wegbreit:cacm75} describes a system for analyzing
simple Lisp programs that produces closed forms that bound
running time.  An interesting aspect of this system is that it
is possible to describe probability distributions on the input
domain (e.g., the probability that the head of an input list
will be some specified value), and the generated bounds incorporate
this information.
\citet{rosendahl:auto_complexity_analysis} proposes a system
based on step-counting functions and abstract interpretation
for a first-order subset of Lisp.  More recently the
COSTA project (see, e.g., \citet{albert-et-al:cost-analysis-java})
has focused on automatically computing resource bounds for
imperative languages (actually, bytecode).

\citeauthor{lematayer:toplas88}'s \citeyearpar{lematayer:toplas88}
ACE system is a two-stage system
that first converts FP\footnote{I.e., Backus' language 
FP \cite{backus:fp}.} programs to recursive FP programs describing 
the number of recursive calls of the target program, then
attempts to transform the result using various program-transformation
techniques to obtain a closed form.
\citet{jost-et-al:popl10} describe a formalism for automatically
infering linear resource bounds on higher-order programs,
provided of course that such bounds are correct, and
\citet{hoffmann-hofmann:esop10} extend this work to handle
polynomial bounds.
This system involves a type system analogous to the target-language
system, but in which the types are annotated with variables
corresponding to resource usage.  Type inference in the annotated
system comes down to solving a set of constraints among these
variables.

Here let us delve a little more into the systems
mentioned in the introduction, and just those aspects concerned
with a general analysis of higher-type call-by-value languages.
\citet{shultis:complexity} defines a denotational semantics for
a simple higher-order language that models both the value and the
cost of an expression.  As a part of the cost model, he develops a
system of ``tolls,'' which play a role similar to the potentials
we define in our work.  The tolls and the semantics are not used
directly in calculations, but rather as components in a logic for
reasoning about them.
\citet{sands:thesis} puts forward a translation scheme in which
programs in a target language are translated into programs \emph{in
the same language} that
incorporate cost information; several target languages are discussed,
including a higher-order call-by-value language.  Each identifier~$f$
in the target language is associated to a \emph{cost closure} that incorporates
information about the value $f$ takes
on its arguments; the cost of applying~$f$ to arguments; and arity.
Cost closures are intended to address the same issue our higher-type
potentials do:  recording information about the future cost
of a partially-applied function.
\citet{van-stone:thesis} annotates the operational semantics for
a higher-order language with cost information.  She then
defines a category-theoretic denotational
semantics that uses ``cost structures'' (which are related to monads) 
to capture cost information and shows
that the latter is sound with respect to the former.
\citet{benzinger:tcs04} annotates NuPRL's call-by-name operational semantics
with complexity estimates.  The language for the annotations is left
somewhat open so as to allow greater flexibility.  The analysis of
the costs is then completed using a combination of NuPRL's proof
generation and Mathematica.
Benzinger's is the only system to explicitly involve automated theorem
proving, though Sand's could also do so.  

These last formalisms are closest to ours in approach, but
differ from ours in a key respect:  the cost domain incorporates information
about values in the target domain so as to provide exact costs,
whereas our approach focuses on upper bounds on costs in terms of
input size.  We are hopeful that our system proves amenable to analyzing
complex programs, but there is much work yet to be done.

\section{Conclusions and further work}
\label{sec:concl_further_work}

We have described a static complexity analysis for a higher-order language
with structural list recursion that yields
an upper bound on the evaluation cost of any typeable program in the
target language.  It proceeds by translating each target-language program~$t$
into a program~$\trans t$ in a complexity language.  
We prove a Soundness Theorem
for the translation that has as a consequence that the cost component
of~$\trans t$ is an upper bound on the evaluation cost of~$t$.
By formalizing the translation and proof of the Soundness Theorem
in Coq, we obtain a machine-checkable certification of that upper
bound on evaluation cost.

The language described here supports only structural recursion on lists;
an obvious extension would be to handle general recursion.  This
should be straightforward if we require the user to supply
a proof of termination of the program to be analyzed.  However,
it should be possible to define
the operational semantics of the target language co-inductively
(as done by, e.g., \cite{leroy-grall:ic09-coind-big-step}), thereby
allowing explicitly for non-terminating computations.  The complexity
language semantics would then have to be adapted so that the
denotation of a recursive complexity function may be partial; the
foundations for such denotational semantics have already been
carried out by
\citet{paulin-mohring:den-semantics-coq} and
\citet{benton-et-al:domain-theory-coq}.  Indeed, one could
then hope to prove termination by extracting complexity bounds and
then proving that these bounds in fact define total functions.

Although we have adapted our formalism to a few other inductively-defined
datatypes (e.g., binary trees), we suspect that there are hidden difficulties
in generalizing the system to handle arbitrary inductive definitions.
One such difficulty might be, e.g., binary trees in which the external
nodes are labeled by values of another inductively-defined datatype.
If we wish to consider both trees and labels as contributing to 
size (potential), then it seems that only external nodes labeled
by size-$0$ values have potential~$0$.  That in turn
may make it difficult to develop useful cost bounds for functions
that ignore the labels.  
\citet{jost-et-al:popl10} deal with this issue by annotating the
type constructors with resource information, and hence automatically
account for the resource information for all types of objects
reachable from the root of a given value.  We have investigated
a similar approach in our setting, in which potential types
mirror the target language types.  For example, we would define
$\typot{\lstinline!$\sigma\,$ list!}$ to be
essentially \lstinline!$\typot\sigma$ list!.
This becomes rather burdensome in practice, and a mechanism for minimizing
the overhead when the generality is not desired would be a necessity.

Another obvious direction would be to handle
different evaluation strategies and notions of cost.
Compositionality is
a thorny issue when considering call-by-need evaluation and lazy
datatypes, and it may be that amortized cost is at least
as interesting as worst-case cost.
\citet{sands:thesis}, \citet{van-stone:thesis},
and \citet{danielsson:popl08} address laziness in their work.
The call-by-push-value paradigm \citep{levy:tlca99} gives
an alternative perspective on our complexity analysis. Call-by-push-value
disinguishes values from computations in a monadic-like approach
under the maxim ``a value is, a computation does.'' With this in
mind, we might adopt the following statement with respect to complexities:
``potential measures what is, cost measures what happens.'' An alternative
presentation of our work might utilize a call-by-push-value target language
to emphasize the distinction between computation expressions
and value expressions and what those mean for the complexity analysis. 

The development in Coq currently works directly with the expressions
of the form~$\trans t$.  These are moderately messy, as
can be seen from the examples.  It
would be nice to provide a more elegant presentation of these complexities
along the lines as the discussion of the examples.
In the same vein, it would be very useful to develop tactics that
allow users to transform complexities into simpler ones, all the while
ensuring the appropriate bounds still hold.
\citet{benzinger:tcs04} addresses this idea, as do
\citet{albert-et-al:jar11} of the COSTA project.
Another relevant
aspect of the COSTA work is that their cost relations use non-determinism
where we have used a maximization operation to handle conditional
constructs; it would be very interesting to see how their approaches
play out in our context.
These tactics could be then applied manually or automatically according
to the user's preferences using Coq's built in automation tools.
Ultimately we should have
a library of tactics for transforming the recurrences produced by
the translation function to closed (possibly asymptotic) forms when possible.

\subsection*{Acknowledgments}
We would like to thank the referees for several helpful suggestions
including pointers to additional related work.

\bibliographystyle{abbrvnat}
\bibliography{semi_auto}

\end{document}